\documentclass[11pt,runningheads,A4]{llncs}
\usepackage{graphicx}
\usepackage{fullpage}
\usepackage{amsmath}
\usepackage[usenames,dvipsnames]{xcolor}
\usepackage[colorlinks,citecolor=blue,linkcolor=BrickRed]{hyperref}
\usepackage{makeidx} 
\usepackage{algorithm}
\usepackage{algorithmic}

\usepackage[margin=2.97cm]{geometry}

\usepackage{tipa}
\usepackage{arcs,lmodern,fix-cm}
\usepackage{times}

\usepackage{amsfonts,latexsym,epsfig,amssymb,color}
\usepackage{mathdots,amsmath,amstext,setspace,enumerate}
\usepackage{slashbox,multirow}
\usepackage{rotating}
\usepackage{verbatim}
\usepackage{caption}
\usepackage{subcaption}

\usepackage{fourier}

\newcommand{\naturals}{\mathbb{N}}

\newcommand{\reals}{\mathbb{R}}

\newcommand{\coss}[1]{\cos\left(#1\right)}
\newcommand{\arccoss}[1]{\arccos\left(#1\right)}
\newcommand{\sinn}[1]{\sin\left(#1\right)}

\newcommand{\ignore}[1]{}

\newcommand{\norm}[1]{\left\lVert#1\right\rVert}

\newcommand{\wsdisk}[1]{$\textsc{WS}_{\textsc{disk}}^{#1}$}
\newcommand{\wsngon}[2]{$\textsc{WS}_{#2\textsc{-gon}}^{#1}$}

\begin{document}
\title{Weighted Group Search on the Disk \\
\& Improved Lower Bounds for Priority Evacuation
\thanks{Research supported in part by NSERC.}
}
%
%
\author{Konstantinos Georgiou\inst{1} 
\and Xin Wang\inst{1}
}
\authorrunning{K. Georgiou and X. Wang}
%
\institute{
Department of Mathematics, Toronto Metropolitan University, Toronto, ON, Canada
\email{\{konstantinos,x85wang\}@torontomu.ca} 
}

\maketitle              

\vspace{-.8cm}

\begin{abstract}
We consider \emph{weighted group search on a disk}, which is a search-type problem involving 2 mobile agents with unit-speed. The two agents start collocated and their goal is to reach a (hidden) target at an unknown location and a known distance of exactly 1 (i.e., the search domain is the unit disk). The agents operate in the so-called \emph{wireless} model that allows them instantaneous knowledge of each others findings.  The termination cost of agents' trajectories is the worst-case \emph{arithmetic weighted average}, which we quantify by parameter $w$, of the times it takes each agent to reach the target, hence the name of the problem. 

Our work follows a long line of research in search and evacuation, but quite importantly it is a variation and extension of two well-studied problems, respectively. The known variant is the one in which the search domain is the line, and for which an optimal solution is known. Our problem is also the extension of the so-called \emph{priority evacuation}, which we obtain by setting the weight parameter $w$ to $0$. For the latter problem the best upper/lower bound gap known is significant. 

Our contributions for weighted group search on a disk are threefold. 
\textit{First}, we derive upper bounds for the entire spectrum of weighted averages $w$. Our algorithms are obtained as a adaptations of known techniques, however the analysis is much more technical. 
\textit{Second}, our main contribution is the derivation of lower bounds for all weighted averages. This follows from a \emph{novel framework} for proving lower bounds for combinatorial search problems based on linear programming and inspired by metric embedding relaxations. 
\textit{Third}, we apply our framework to the priority evacuation problem, improving the previously best lower bound known from $4.38962$ to $4.56798$, thus reducing the upper/lower bound gap from $0.42892$ to $0.25056$.

\end{abstract}
%
%
%




\section{Introduction}
Autonomous mobile agent searching over geometric domains such as lines, disks, circles,  triangles and polygons has been the subject of extensive research over the last decades. In these problems, a fleet of agents is tasked with finding a hidden item in the geometric domain while complying with searchers' specifications. One of the most significant parameters that distinguish these variations is how the cost of the solution is quantified, effectively identifying the set of feasible solutions and changing the computational boundaries of the underlying problem. 
Indeed, in traditional search problems, e.g. \emph{search} and rescue scenarios, one is concerned just with the finding of the hidden item, hence quantifying the cost as the time it takes the first agent to reach the target. In the other extreme, e.g. in \emph{evacuating} scenarios, one is concerned with minimizing the time the last agent reaches the hidden item. 
These objective variations are combined with other fundamental specifications, such as the communication model and possibility of faultiness. 

Despite the growing number of problems in the field, the number of problems for which matching upper and lower bounds are known is also increasing. Empirically, problems that admit full symmetry, involving the agents' specifications, domain, communication model, etc., are those that admit strong or matching upper and lower bounds. However, even a small deviation from symmetry makes any lower bounds to the problems particularly difficult to tackle. Examples include the face-to-face communication model, where knowledge is shared in an asymmetric way during the execution of the search, different agent speeds, or abstract domains such as triangles. Recently, asymmetry was also introduced as part of the way the search cost is quantified by considering the objective of only a \emph{distinguished agent} reaching the hidden item, a problem known as \emph{priority evacuation}.
Negative results for this and similar problems are rare, and usually weak and 
challenging. 

The main contribution of our work is a framework that provides lower bounds for combinatorial search problems. More generally, our starting point is the study of a generalization of the priority evacuation problem on the disk, which we obtain by considering \emph{asymmetric} cost functions. Aside from the upper bounds we derive for these problems, our primary contribution is a framework for proving lower bounds (for these as well as for more general problems) that is based on linear programming and metric embedding relaxations. We demonstrate the usefulness of the framework by establishing lower bounds for our weighted evacuation problems, and in particular by improving the previously best lower bound known for the priority evacuation problem. 

\ignore{
For this family of problems we obtain upper bounds.
search cost is the arithmetic weighted average of the times the agents reach the hidden item, with a weight equal to $0$ corresponding to priority evacuation. 
Our results are multi-fold. Firstly, we provide a family of search algorithms that adapt the agents' trajectories together with the weights of the weighted average. Secondly, and most importantly, we introduce a new framework for proving lower bounds capable of tackling problems lacking symmetry as it arises for example by the cost function.
Our framework relies on the solution of linear programs (inspired by metric embedding relaxations) obtained as relaxations to the exact formulations to optimal solutions to search problems for specific search domains with discrete and finite many locations for the hidden item. Thirdly, and as a byproduct of our results, we derive an improved lower bound 
to the priority evacuation. 
}

\vspace{-.2cm}

\subsection{Related Work}


The study of search-type problems dates back to the 1960's and has resulted to a rich theory that has been summarized in books~\cite{alpern2013search,AlpGal03} and surveys~\cite{czyzowicz2019groupkos,11340,hohzaki2016search}, among others. 
Originally, the problem pertained to the identification of a hidden object within a search domain, and hence the objective was to minimize the time the object was found~\cite{beck1964linear}. 
Already with one searcher, the problem has seen interesting, and surprisingly challenging variations touching on the searcher's specs and the search domain, some of which we briefly discuss next. More recently, and with the emergence of fleet-robotics, search-type problems  were considered with multiple searchers, see e.g.~\cite{CzyzowiczGGKMP14}.

A basic example of a search space is the so-called linear search problem that has been considered with one~\cite{baezayates1993searching} or multiple searchers~\cite{ChrobakGGM15}. The same problem has also been considered with the objective of minimizing the weighted average of the searchers' termination times~\cite{GLweightedLine2023} (that we also use in this work). Other 1-dimensional settings include variations such as searching rays~\cite{BrandtFRW20} or graphs~\cite{AngelopoulosDL19},
or also search for multiple objects~\cite{Borowiecki0DK16,CzyzowiczDGKM16}.

Searching two-dimensional domains has become a more dominant topic in the last decade or so. 
Considered domains include polygons~\cite{FeketeGK10}, the quite popular disk~\cite{CzyzowiczGGKMP14}, 
the plane~\cite{feinerman2017ants},
regular polygons~\cite{czyzowicz2020priority123}
the equilateral triangle and square~\cite{ChuangpishitMNO20,CzyzowiczKKNOS15}, 
arbitrary triangles~\cite{georgiou2022triangle}, and $\ell_p$ unit disks~\cite{GLLKllp2023}. 
Many of these problems have also seen variations pertaining to the communication model or more generally to the  the searchers' specifications, 
e.g. searching the disk in the face-to-face model~\cite{disser2019evacuating},
with different searchers' speeds~\cite{BampasCGIKKP19},
and with different searchers' communication capabilities~\cite{czyzowicz2021groupevac,GGK2022asym}.
Last but not least, search problems have also been considered under faultiness settings, see for example
\cite{BGMP2022pfaulty,CzyzowiczGGKKRW17,czyzowicz2021searchbyz,czyzowicz2021searchnew,CzyzowiczKKNO19,GeorgiouKLPP19,Sun2020}

A number of search problems have also been considered with less standard objectives.
For example~\cite{chuangpishit2020multi} considered a multi-objective search-type problem, 
\cite{0001DJ19} studied search problems under a broad competitive algorithmic lens, 
\cite{MillerP15} considered information/cost trade-offs, 
\cite{czyzowiczICALP,czyzowicz2021energy} considered time/energy trade-offs,
and \cite{kranakis2019search} introduced search-and-fetch problems in two dimensions. 
More closely related to our work, is the so-called priority evacuation objective introduced in~\cite{czyzowicz2020priority123,czyzowicz2020priority4} where search termination is called when a distinguished searcher reaches the target.



\subsection{Discussion on New Results}
In this study, we explore a natural extension of the priority evacuation problem on the disk introduced in~\cite{czyzowicz2020priority123,czyzowicz2020priority4}. Our specific problem, termed the \emph{weighted group search} on the disk, is parameterized by $w\in[0,1]$. Parameter $w$ represents a designated arithmetic weighted average (i.e. $\tfrac{x+wy}{1+w}$, when considering the weighted average of $x,y$) that defines the objective function, hence the name of the problem. The same objective was previously considered with the line as the search domain and solved optimally in~\cite{GLweightedLine2023}.
Moreover, our weighted group search problem with $w=0$ corresponds exactly to the previously studied priority evacuation problem on the disk for which the best lower and upper bound known are $4.38962$ and $4.81854$, respectively, creating a notable gap that remains an open problem. 

Our results are derivations of upper and lower bounds for the weighted group search problem on the disk for all $w \in [0,1]$. Our upper bounds are obtained by a family of mobile agent trajectories (algorithms) that are adapting with $w$. The derived algorithm for $w=0$ is exactly the best algorithm known for the previously studied priority evacuation problem (with performance $4.81854$), and therefore our positive results can be understood as generalizations, still technical to analyze, of known algorithmic techniques. 
The main contribution of this work is the introduction of a novel framework for deriving lower bounds for general search problems. 
We discuss the key technical and conceptual ideas in the following section. Here we emphasize the significance of the obtained results in an area where finding strong lower bounds is notoriously challenging and rare.
Indeed, we design a framework for proving lower bounds that is applicable to general objective functions, and to general discrete search domains. Our framework utilizes Linear Programming relaxations that are inspired by metric embedding relaxations of Euclidean 2-dimensional metric spaces induced by optimal search strategies. 
Indeed, we apply our framework to the weighted group search problem on the disk and we find lower bounds for all $w \in [0,1]$, resulting in upper/lower bound gaps that are diminishing with $w$. 
When $w=1$ the bound is optimal. However, the punchline of our new methodology pertains to the lower bound obtained for $w=0$,  corresponding to previously studied priority evacuation. For this problem our techniques allow us to improve the previously best lower bound known to $4.56798$, reducing this way the upper/lower bound gap from $0.42892$ to $0.25056$. 

\subsection{Key Technical and Conceptual Ideas}
Our approach to establishing \emph{upper bounds} for weighted group search involves a generalization of the algorithm presented in~\cite{czyzowicz2020priority123}. We achieve this by specifically parameterizing the algorithm's behavior and performance according to the parameter $w$ of the cost function. Unsurprisingly, when $w=0$, our algorithm aligns with the one detailed in~\cite{czyzowicz2020priority123}.
The inherent challenge in analyzing our family of algorithms, which vary with $w$, lies in the fact that they generate diverse trajectories for the agents, where critical domain points are visited with varying order. These points are possible locations of the hidden item, which have been explored already, but whose close neighborhood has not been explored yet, and hence a adversarial placement of the hidden item arbitrarily close to those points is a possibility. Such points are necessarily introduced by strong search algorithms, as in these algorithms an agent may abandon the search of the disk in order to move in its interior, potentially in order to expedite evacuation should the hidden item be located by the other agent. 
Hence, identifying the optimal choices within that family (as a function of $w$) becomes more technical than the case $w=0$. 
Consequently, quantifying the solution cost for all inputs necessitates a meticulous analysis, which at a high level considers various searched space configurations. 
For this reason, the contribution of this section is primarily only technical in nature.

The main contribution of this work lies in introducing a novel framework for proving \emph{lower bounds} applicable to various combinatorial search-type problems. As an illustration of our methodology, we derive lower bounds for the weighted group search problem across all values of $w$. Specifically, we establish a matching lower bound for $w=1$ and, notably, an improved lower bound for $w=0$ addressing the priority evacuation problem.
The challenge inherent in proving lower bounds for the weighted group search (as well as the priority evacuation or similar problems) stems from the contiguous search domain and the asymmetry of the objective function. Any lower bound argument necessitates the consideration of snapshots of an algorithm at certain time-stamps, identifying a finite collection of potential input placements for which the algorithm must perform well. Additionally, due to the online nature of the problem, the algorithm must operate identically until the input is discovered.
The critical question revolves around how these finite points are processed by the two agents. A key observation is that the time-stamps, combined with potential input placements and agent trajectories, induce a finite 2-dimensional $\ell_2$ metric space. By conditioning on the order in which input placements are visited, one can model the optimal algorithm's performance using a Non-Linear Program (NLP). While solving all NLPs (corresponding to all permutations of placement visitations) identifies the optimal algorithm for the specific finite search domain, the challenge lies in the difficulty of solving these NLPs and providing proper certificates of optimality.
In that direction, a second key observation is that one can relax the induced 2-dimensional $\ell_2$ metric spaces to abstract metric spaces, leading to Linear Programs (LPs). These LPs can be solved efficiently, with accompanying certificates of optimality. Importantly, the optimal values of these LPs serve as lower bounds for the NLPs, consequently establishing lower bounds for the original problems.
Our framework relies among others on symbolic (not numerical), computer-assisted calculations, which we did not exhaust when deriving the lower bounds, suggesting that further improvements are possible by bypassing current computational limitations.

\subsection{Roadmap}
We start with Section~\ref{sec: problem definition notation} where we give some formal definitions and introduce proper terminology. This is followed by Section~\ref{sec: past results} where we present past and new results using the underlying new terminology, along with some key observations. Section~\ref{sec: new framework} introduces our first main contribution, which is a framework based on (metric-inspired) linear program relaxations that give rise to lower bounds for geometric search-type problems. 
In Section~\ref{sec: implied LB priority evacuation} we apply our framework to the previously studied priority evacuation problem with 2 agents, effectively improving the previously best lower bound known. 
Finally, in Section~\ref{sec: arithmetic weighted disk} we study the weighted group search problem on the disk, which is a generalization of the priority evacuation problem, deriving upper and lower bounds, the latter using our general framework of Section~\ref{sec: new framework}.
Finally, in Section~\ref{sec: conclusions} we conclude with some future directions. 

\section{Preliminaries}

\subsection{Problem Definition, Notation \& and some Observations}
\label{sec: problem definition notation}

We consider a class of search-type problems $\textsc{WS}_{\mathcal{D}}^{f}$ (for Wireless Search) with 2 agents in the wireless model, which we define next. 
In these problems, 2 unit speed agents are initially collocated at the origin of the Euclidean plane, i.e. $\reals^2$ equipped with the $\ell_2$ metric. The 2 agents have distinct (known) identities that we call $A_0, A_1$. Given a known geometric object $\mathcal{D}\subseteq \reals^2$ and a known cost function $f:\reals^2 \mapsto \reals$ (non-decreasing in both coordinates), a solution to the problem is given by trajectories $\tau_i:\reals_{\geq0} \mapsto \reals^2$, $i=0,1$, that induce movements for the 2 agents of speed at most 1. 
The two agents operate in the \emph{wireless model} and they need to find a hidden target, in the following mathematical sense. 
For each fixed \emph{target} $I \in \mathcal{D}$, we allow \emph{both} trajectories $\tau_i$ to depend on $I$ \emph{only after} $I$ is visited by \emph{any} agent, i.e. when for some $t\geq0$ and $i\in \{0,1\}$ we have that $\tau_i(t)=I$. 
Equivalently, two executions of the trajectories are identical for different inputs up to the moment the first of these inputs is hit by a trajectory. 

For each $I \in \mathcal{D}$, we also denote the \emph{termination time} $T_i(I)$ of agent $A_i$ as the first time that agent $A_i$ visits target $I$.
The agent who is the first that visits the target is referred to as the \emph{finder}. We emphasize that each $T_i(I)$ depends on all $\tau_0,\tau_1$ and $I$ (while trajectories $\tau_j$ may only depend on $I$ after the target is visited). 
The objective of $\textsc{WS}_{\mathcal{D}}^{f}$ is to determine trajectories $\tau_0,\tau_1$ so as to minimize the \emph{search cost}
$$
\sup_{I \in \mathcal D} 
\tfrac{ f\left( T_0(I), T_1(I) \right) }
{f\left( \|I\|_2, \|I\|_2 \right) }.
$$
When $f$ is the $\max$ function, the problem is known as \emph{evacuation}, and when $f$ is the projection function, then it is known as \emph{priority evacuation}. For this reason, also refer to the objective also as the \emph{evacuation cost}. 
In this work we consider the following search domains: (a) $\mathcal{D} = \textsc{disk}$, the unit radius disk, and (b) $\mathcal{D} = n$-\textsc{gon}, the vertices of a regular $n$-gon inscribed in the unit radius disk.
\ignore{
\begin{enumerate}[(a)]
\item 
$\mathcal{D} = \textsc{disk}$, the unit radius disk 
\item $\mathcal{D} = n$-\textsc{gon}, the vertices of a regular $n$-gon inscribed in the unit radius disk 
\end{enumerate}
}

Some important observations are in place. 
\emph{First}, the underlying search space (i.e. the space where agents' movements take place) is the 2-dim Euclidean space. 
Searching in other $\ell_p$ metrics has been considered in~\cite{GLLKllp2023}. 
For convenience, we think of the search domains embedded on the Cartesian plane so that the unit disks are centered at the origin. The fact that the underlying search space is a metric space, and in particular the Euclidean metric space will be essential in our lower bound arguments. 
\emph{Second}, for a given $I \in \mathcal D$, let agent $A_i$ be the finder of the target, and suppose that this happens at time $t$, i.e. $I=\tau_i(t)$. Since the other agent can have her trajectory depend on this finding (wireless model), we may assume that $A_{1-i}$ moves directly to the target, i.e. that 
\begin{equation}
\label{equa: last jump}
T_{1-i}(I) = t+\norm{\tau_{i}(t) - \tau_{1-i}(t)}_2= t+\norm{I - \tau_{1-i}(t)}_2.
\end{equation}
\emph{Third}, for the search domains we consider, i.e. the disk or the vertices of $n$-gons, all possible targets are at distance $1$ from the origin. 
Agents that knew in advance the position of any target $I$ (i.e. if trajectories $\tau_i(t)$ could depend on $I$ for all $t\geq0$)
would need time $1$ to reach the target, inducing cost $f\left( \|I\|_2),\|I\|_2)  \right)= f(1,1)$, independent of the target $I$. 
Hence, when quantifying the performance of a search trajectory as per its search cost,
we perform both worst case and competitive analysis. In particular, both performance quantifications admit the same optimal trajectories, and the corresponding optimal search costs are off by constant multiplicative factor $f(1,1)$ that depends only on $f$, hence they are also equal when $f(1,1)=1$. 

\subsection{Past Results, Search Domains \& Cost Functions}
\label{sec: past results}

A number of past results can be described in the framework of $\textsc{WS}_{\mathcal{D}}^{f}$ problems as we demonstrate next. The discussion focuses on the wireless model. For other communication models, one has to adjust the definition of feasible trajectories. 

Typical search problems, where search is complete when the first agent reaches the target, are associated by definition with cost function $f(x,y)=\min\{x,y\}$. 
When multiple agents are involved in search, and one quantifies the performance by the time the last agent reaches the target, then one uses $f(x,y)=\|(x,y)\|_\infty=\max\{x,y\}$. The case where the performance is determined by the termination time of a designated agent, then one needs the projection function $f(x,y)=\textrm{proj}_2(x,y)=y$ (or $\textrm{proj}_1$). 
Finally, one may also consider a weighted average of the termination times of the two agents, by using cost function $f(x,y)=g_w(x,y):=w x + y$, where without loss of generality $w\in [0,1]$ (note that one may use a scaling factor $w+1$  without affecting optimizers). With these definitions in mind, we have that 
$
g_0 = \textrm{proj}_2
$
and $g_1 = \| \cdot \|_1$.

There are four notable examples of $\textsc{WS}_{\mathcal{D}}^{f}$ problems 
that were considered before. Our new findings build upon ideas found in these results, which we naturally generalize as well as in some cases we also improve:
\begin{itemize}
\item $\textsc{WS}_{\reals\setminus [-1,1]}^{g_w}$, also known as weighted group search on a line, was considered in~\cite{GLweightedLine2023} where the search space $\mathcal D$ was the one dimensional real line excluding all points at most $1$ away from the origin. The problem has been solved optimally. Our work is the first to consider the same cost function but on the disk.
\item $\textsc{WS}_{\textsc{disk}}^{\|\cdot\|_\infty}$ is the classic evacuation problem on the disk, solved optimally in~\cite{CzyzowiczGGKMP14}, where the cost function is simply described as the time that the last agent reaches the target.  
\item $\textsc{WS}_{\textsc{disk}}^{\textrm{proj}_2}$ is the priority 
evacuation problem (with 1 servant and 1 distinguished searcher, the queen) considered in~\cite{czyzowicz2020priority123}. In our current work, first we improve the lower bound for the problem, as well as we generalize it by considering cost functions 
$
g_w
$, where $w\in [0,1]$, and the same search domain. Recall that $
g_0 = \textrm{proj}_2
$.
\item \wsngon{\textrm{proj}_2}{6} was considered in~\cite{czyzowicz2020priority123} and solved optimally. Our methodology implies the same lower bound. Moreover, we are the first to study search domain $n$-\textsc{Gon}, with $n>6$, where in particular we derive lower bounds for \wsngon{g_w}{n}, $w\in [0,1]$, for $n=7,8,9$. 
\end{itemize}
The motivation for studying search domain $n$-\textsc{Gon} 
relates to the fact that a lower bound to the problem is also a lower bound to searching 
domain \textsc{disk} (for the same cost function), 
an idea introduced in~\cite{CzyzowiczGGKMP14}, and later used in~\cite{czyzowicz2020priority123}. 
The following lemma is explicit in these works. Even though it was previously used only for the $\|\cdot \|_\infty$ and $\textrm{proj}_2$ cost functions, it holds more generally for any non-decreasing cost function $f$. Here we generalize the statement, since we will need it when studying the weighted group search problem. We emphasize that the lemma establishes lower bounds to algorithms addressing the search of the disk, conditioning on their performance of searching $n$-gons. 

\begin{lemma}
\label{lem: ngone bound implies disk bound}
Let $t_0$, $t_1$ be the termination time lower bounds of agents $A_0,A_1$, respectively, for some input to \wsngon{f}{n}.
Then, no algorithm for \wsdisk{f} with these termination times of the agents
has evacuation cost better than
$
\frac{ f\left( t_0+\pi/n, t_1+\pi/n  \right) }
{f\left( 1,1 \right) }.
$
\end{lemma}

\begin{proof}
We let an arbitrary algorithm for \wsdisk{f} run for time $1+\pi/n-\epsilon$. In this time, the 2 agents have searched at most $2\pi/n-2\epsilon$ of the disk. It is easy to see then that there is an inscribed regular $n$-gon, none of whose vertices have been explored yet. Note that \textsc{$n$-gon} is a subset of the search domain \textsc{disk}, and hence any target $I$ of \textsc{$n$-gon} is an eligible target of \textsc{disk}. For target $I$, the agents' termination times are increased by $\pi/n$, and the claim follows. 
\qed\end{proof}

Specifically all cost-functions $f \in \{\textrm{proj}_2, \|\cdot \|_\infty, g_w\}$ have the property that 
$
f(x+a,y+a)/f(1,1)=f(x,y)/f(1,1)+a.
$ 
Now in Lemma~\ref{lem: ngone bound implies disk bound}, let $c_n$ be the derived lower bound for \wsngon{f}{n} induced by target $I$. 
The same target $I$ is an eligible adversarial choice for \wsdisk{f}, hence for the latter problem, the target is reached at least $\pi/n$ time later. This means, the following quantity is a valid lower bound to \wsdisk{f}. 
\begin{equation}
\frac{ f\left( T_0(I)+\frac{\pi}n, T_1(I)+\frac{\pi}n  \right) }
{f\left( 1,1 \right) }
=
\frac{f\left( T_0(I),T_1(I)\right)}{f(1,1)} +\frac{\pi}n
= c_n +\frac{\pi}n
\label{equa: pseudo linear}
\end{equation}
In fact the best lower bound of $\pi/6+3+\sqrt3/2$ for \wsdisk{\textrm{proj}_2} is due to the provable lower bound of $3+\sqrt3/2$ for \wsngon{\textrm{proj}_2}{6}, see~\cite{czyzowicz2020priority123}. 
In this work we strengthen this result by deriving \textit{new lower bounds} for \wsngon{\textrm{proj}_2}{n}, $n>6$.

\section{Improved Framework for Proving Lower Bounds}
\label{sec: new framework}

In this section we leverage the existing framework for showing lower bounds for \wsngon{f}{n}, and in light of Lemma~\ref{lem: ngone bound implies disk bound}, that would imply adjusted lower bounds for \wsdisk{f}, too. In what follows $[n]$ denotes the set $\{0,1,\ldots, n\}$, so for example $[1]=\{0,1\}$. Moreover, we denote the set of all permutations of $[n]$ by $\mathcal R_n$.

Our first task is to provide a systematic way in order to find the optimal solution to \wsngon{f}{n}. At a high level that would be accomplished by solving $\frac{(n-1)!2^{n-1}}{n}$ many Non-Linear Programs. 
Recall that we think of the \textsc{$n$-gon} embedded on the Cartesian Euclidean space, inscribed in a unit-radius circle and centered at the origin. For this reason its vertices are identified by points 
$
\left(\coss{2i\pi/n}, \sinn{2i\pi/n}\right)$, $i=0,\ldots,n-1$.
Now, for a permutation $\rho \in \mathcal R_n$ (corresponding to the $n$ vertices of \textsc{$n$-gon}), we slightly abuse notation and we write $\rho_i$ to denote both the $i$-th element of the permutation, as well as the corresponding point $\left(\coss{2\rho_i \pi/n}, \sinn{2\rho_i \pi/n}\right)$ on the plane. 

The main idea of our formulation is that search strategies for \wsngon{f}{n} can be classified with respect to the order that vertices are visited, given by some permutation $\rho$ over $[n]$, and the identities of the agents that visit these vertices, given by a binary string $b\in [1]^n$. In particular, this means 
vertex $\rho_i$, i.e. the $i$-th visited vertex, is visited no later than vertex (target) $\rho_j$, when $i<j$, and that target $\rho_i$ is visited by agent $A_{b_i}$. We call such a search strategy a $(\rho,b)$-algorithm.  
Note that without loss of generality (due to symmetry), we may assume that the first vertex to be visited is vertex 0, and that the second visited vertex is one among $1, \ldots, \lfloor (n-1)/2 \rfloor +1$. This gives rise to at most $\frac{(n-1)!2^{n-1}}{n}$ classes of search algorithms. Next we show how to find the optimal $(\rho,b)$-algorithm, for each fixed $\rho,b$, by solving a Non-Linear Program.

\begin{lemma}
\label{lem: exact formulation}
For each $n \in \naturals$, permutation $\rho$ of $[n]$, binary string $b\in [1]^n$ and $t_0 \in \reals_{\geq 0}$, consider the Non-Linear Program\footnote{For the intended meaning of the variables, the reader may consult the proof of the lemma.}
\begin{align}
\min &~\max_{i \in [n]} 			
\left\{  
f\left(c^0_i,c^1_i\right)
\right\}   
\tag{$\textsc{NLP}^f_n(\rho,b,t_0)$}  \label{equa:nlp}   \\
s.t.: ~&
 t_n\geq t_{n-1}\geq \ldots \geq t_1 \geq t_0 \notag \\
& t_{i+1}-t_i \geq \norm{L^j_{i+1}- L^j_{i}}_2, ~~i\in [n-1], j\in [1] \notag \\
& c^j_i \geq t_i + \norm{L^j_{i}- L^{b_i}_{i}}_2, ~~i\in [n], j\in [1] \notag 
\end{align}
in variables $\{c^j_i, t_i, x_i,y_i\}_{j\in [1], i\in [n]}$, where in particular we used abbreviations 
$L^{b_i}_i =\rho_i$ and $L^{1-b_i}_i =(x_i,y_i)$. 
\ignore{
$$
L^j_i 
=
\left\{
\begin{array}{ll}
\rho_i		&, j=b_i \\
(x_i,y_i)	&, j=1-b_i
\end{array}
\right..
$$
}
Then, for every convex cost function $f$, \eqref{equa:nlp} admits a unique (hence global) minimum, 
and the optimizer for $t_0=1$ corresponds to the optimal $(\rho,b)$-algorithm for \wsngon{f}{n}.
\end{lemma}

\begin{proof}
It is easy to see that the domain of the Non-Linear Program is convex. 
Indeed, any constraint of the form $\|\|_2 \leq c$ is convex, hence the intersection of convex constraints define a convex feasible region. 
As for the objective $f$ is given to be convex, while $\max\{\cdot\}$ is convex too, hence also their composition. 
This makes the Non-Linear Program convex. 
Note also that the non-linear program is bounded from below  by $f\left(t_0t_0\right)$, hence it admits an optimal solution, which is the only (global) minimum. 

Next we argue that for the fixed $\rho,b$, and by setting $t_0=1$, the Non-Linear Program correctly finds the termination cost of the optimal $(\rho,b)$-algorithm for \wsngon{f}{n} (scaled by $f(1,1)$). 
For this, we present the semantics of the chosen variables that identify the behavior of the $(\rho,b)$-algorithm. \\
- $c^j_i$ is the time by which agent $j$ reaches vertex $\rho_i$, $j\in[1], i \in [n]$. \\
- $t_i$ is the time by which vertex $\rho_i$ is visited for the first time, by any agent, $i \in [n]$.\\
- $x_i, y_i$ are coordinates of agent $j =1- b_i$, i.e. the agent that does not visit $\rho(i)$ first, $i \in [n]$.

Therefore, indeed, if $L^j_i$ denotes the position of agent $A_j$ when potential target $\rho_j$ is visited
(which are subject to an algorithmic choice), we have that 
\begin{align*}
L^{b_i}_i &=\rho_i, \\
L^{1-b_i}_i &=(x_i,y_i),
\end{align*}
for each $i$. 
By~\eqref{equa: last jump}, we obtain 
$$c^j_i \geq t_i + \norm{L^j_{i}- L^{b_i}_{i}}_2,$$
i.e. a bound for agent $j$ to visit vertex $\rho_i$. 
Since $f$ is non-decreasing in each coordinate, this constraint will be satisfied with equality for the optimal solution (hence we could also eliminate variables $c^j_i$), but we leave the constraint as is for the sake of better exposition. 

Clearly, because the agents are moving at speed at most 1, the time elapsed between the visitation of $\rho_{i+1}$ and $\rho_i$ cannot be less than the time needed by each agent $A_j$ to move from $L^j_{i}$ to $L^j_{i+1}$, i.e. no more than $\norm{L^j_{i+1}- L^j_{i}}_2$. Finally, the time that the first vertex is visited is at least $t_0=1$, overall showing that the constraints capture the $(\rho,b)$-algorithm visitation times of all potential targets. 

Finally, as per the definition the search cost,
and due to that in \wsngon{f}{n} there are only $n$ many potential targets, one needs to compute the search cost for each potential target $\rho_i$, which equals $f\left(c^0_i,c^1_i\right)/f(1,1)$, only that in the description of the non-linear program we omit the multiplicative factor $f(1,1)$. 
\qed\end{proof}

We slightly abuse notation, and we denote by $\textsc{NLP}^f_n(\rho,b,t_0)$ also the optimal value of the same Non-Linear Program. 
The following corollary follows immediately by our definitions, and Lemma~\ref{lem: exact formulation}. 
\begin{corollary}
\label{cor: optimal n-gone}
The optimal search cost for \wsngon{f}{n} equals
$
\frac{1}{f(1,1)} 
\min_{\rho \in \mathcal R_n, b \in [1]^n }
\textsc{NLP}^f_n\left( \rho,b,1 \right).
$
\end{corollary}

Note that any solution to~\eqref{equa:nlp} is associated with an embedding of points $L_i^j$ in the $(\reals^2, \ell_2)$ metric space, where in particular 
$
d(L_i^j, L_{i'}^{j'}) = \norm{L_i^j - L_{i'}^{j'}}_2.
$
This metric space satisfies the triangle inequality. Therefore, requiring that distances 
$
d_{j,i,j',i'}:=d(L_i^j, L_{i'}^{j'})
$
satisfy the triangle inequality, but not necessarily that the space is embeddable to $\reals^2$ (or even that it is $\ell_2$), gives rise to a natural relaxation to the problem. This idea is materialized in the next lemma.

\begin{lemma}
\label{lem: lb to n-gone}
For each $n \in \naturals$, permutation $\rho$ of $[n]$, $b\in [1]^n$ and $t_0 \in \reals_{\geq 0}$, consider the Non-Linear Program (Relaxation)
\begin{align}
\min &~\max_{i \in [n]} 			
\left\{  
f\left(c^0_i,c^1_i\right)
\right\}   
\tag{$\textsc{REL}^f_n(\rho,b,t_0)$}  \label{equa:rel}   \\
s.t.: ~&
 t_n\geq t_{n-1}\geq \ldots \geq t_1 \geq t_0 \notag \\
& t_{i+1}-t_i \geq d_{j,i+1,j,i}, ~~i\in [n-1], j\in [1] \notag \\
& c^j_i \geq t_i + d_{j,i,b_i,i}, ~~i\in [n], j\in [1] \notag \\
& d_{b_i,i,b_{i'},i'} = \norm{ \rho_i - \rho_{i'}}_2, ~~i,i' \in [n] \notag \\ 
& \left(\left\{L^j_i\right\}_{j\in [1], i\in [n]}, d(L^j_i,L^{j'}_{i'}))=d_{j,i,j',i'} 
\right) ~\textrm{is a metric} \notag
\end{align}
in variables 
$\{c^j_i, t_i, d_{j,i,j',i'}\}_{j,j'\in [1], i,i'\in [n]}$.
Then,~\eqref{equa:rel} admits a unique (global) minimum, that we denote by $\textsc{REL}^f_n(\rho,b,t_0)$, and in particular
$
\textsc{NLP}^f_n(\rho,b,t_0)
\geq
\textsc{REL}^f_n(\rho,b,t_0).
$
\end{lemma}


\begin{proof}
First we observe that the feasible region of~\eqref{equa:rel} is a polyhedron. 
Indeed, for all $i,i' \in [n]$, we have that $\norm{ \rho_i - \rho_{i'}}_2$ is a constant, while the axioms that distance function $d_{j,i,j',i'}$ defines a metric for points $\{L^j_i\}_{j\in [1], i\in [n]}$ is a collection of linear constraints (including the triangle inequalities).
It is also easy to see that the objective is bounded from below by $f(t_0,t_0)$, hence, the Non-Linear Pogram admits a unique local (global) minimizer. 

Finally, we show that~\eqref{equa:rel} is a relaxation to~\eqref{equa:nlp}, implying the claimed inequality. 
To see why, note that the constraints of~\eqref{equa:nlp} require that  
$$
\left(\left\{ L^j_i\right\}_{j\in [1], i\in [n]}, d(L^j_i,L^{j'}_{i'}))=d_{j,i,j',i'} 
\right)
$$
is embeddable in $(\reals^2, \ell_2)$, and therefore the distance function 
$
d_{j,i,j',i'}=\norm{L^j_i-L^{j'}_{i'}}_2
$
is a metric. 
\qed\end{proof}

The significance of~\eqref{equa:rel} over~\eqref{equa:nlp} is that the former can be solved much faster, especially when $f$ is a linear function (e.g. when $f\in \{\textrm{proj}_2, \|\cdot \|_\infty, g_w\}$), in which case the resulting relaxation is a Linear Program (a basic trick introduces a new linear variable, and a collection of inequality constraints that simulate that the optimal solution simulates the $\max$ function). 
Because~\eqref{equa:rel} is easy to solve, it will be used to provide search cost lower bounds to \wsdisk{f}, as the next lemma suggests. 

\begin{lemma}
\label{lem: lb to disk}
For every $n \in \naturals$, 
no algorithm to \wsdisk{f} has search cost better than 
$$
\mathcal L_n^f 
:=
\frac{1}{f(1,1)} \min_{\rho \in \mathcal R_n, b \in [1]^n }
\textsc{REL}^f_n\left( \rho,b,1+\frac{\pi}n \right).
$$
\end{lemma}

\begin{proof}
By Corollary~\ref{cor: optimal n-gone} and Lemma~\ref{lem: ngone bound implies disk bound}, and for all $n\in \naturals$ we have that 
any solution to \wsdisk{f} has cost at least 
$
\frac{1}{f(1,1)} 
\min_{\rho \in \mathcal R_n, b \in [1]^n }
\textsc{NLP}^f_n\left( \rho,b,1+\frac{\pi}n \right)
\geq 
\frac{1}{f(1,1)}
\min_{\rho \in \mathcal R_n, b \in [1]^n }
\textsc{REL}^f_n\left( \rho,b,1+\frac{\pi}n \right),
$
where the last inequality is due to Lemma~\ref{lem: lb to n-gone}. 
\qed\end{proof}

Lemma~\ref{lem: lb to disk} effectively implies that for each $n$, one can solve 
$(n-1)!2^{n-1}/n$ Convex Programs (or more specifically Linear Programs, if $f$ is linear), in order to obtain lower bound $\mathcal L_n^f$ for \wsdisk{f}. This idea is explored in Section~\ref{sec: implied LB priority evacuation} specifically for $f=\textrm{proj}_2$, and more generally in Section~\ref{sec: arithmetic weighted disk} for cost functions $f=g_w$, $w\in [0,1]$.

\section{Implied (and Improved) Lower Bounds for Priority Evacuation}
\label{sec: implied LB priority evacuation}

In this section we present and prove one of our main contributions, which is an improved lower bound to the Priority Evacuation  Problem  \wsdisk{\textrm{proj}_2}. 

\begin{theorem}
\label{thm: new priority lb}
No algorithm for \wsdisk{\textrm{proj}_2} has evacuation cost less than 
$
1+\pi/9 + \sqrt3/2 + \coss{\pi/18} + 4\sinn{\pi/9} \approx 4.56798.
$
\end{theorem}

As a reminder, the previously best lower bound known for the problem was $4.38962$. Note that not only the improvement is significant, but most importantly, the provided framework allows for even further improvements if one utilizes computational resources more efficiently, e.g. use specialized software for solving LP's instead of \textsc{Mathematica} that was used in this project. 
 
The proof of Theorem~\ref{thm: new priority lb} follows immediately from Lemma~\ref{lem: lb to disk} once we present new lower bounds for \wsngon{\textrm{proj}_2}{n}, $n>6$. 
Indeed, for $f=\textrm{proj}_2$, we have that $f(1,1)=1$ and the cost function is linear. 
For this reason, 
due to Lemma~\ref{lem: lb to disk}, and as per the calculations~\eqref{equa: pseudo linear}, 
no algorithm for \wsdisk{\textrm{proj}_2} has evacuation cost less than 
$
\frac{1}{f(1,1)} \min_{\rho \in \mathcal R_n, b \in [1]^n }
\textsc{REL}^{\textrm{proj}_2}_n\left( \rho,b,1+\frac{\pi}n \right)
=
\frac{\pi}n+
\min_{\rho \in \mathcal R_n, b \in [1]^n }
\textsc{REL}^{\textrm{proj}_2}_n\left( \rho,b,1 \right),
$
where in particular $\textsc{REL}^{\textrm{proj}_2}_n\left( \rho,b,1 \right)$ is a lower bound to the optimal $(\rho,b)$-algorithm for \wsngon{\textrm{proj}_2}{n}.
For each $n=3,\ldots, 9$, we solve $\textsc{REL}^{\textrm{proj}_2}_n(\rho,b,1)$ for all permutations $\rho \in \mathcal R_n$ and all binary strings $b\in [1]^n$, and we report the smallest values in Table~\ref{tab: n-gone lower bounds}. 
The calculations were computer-assisted but also symbolic (non-numerical).\footnote{Calculations were performed symbolically with \textsc{Mathematica}. The solution to any LP comes with a proof of optimality.} 

\begin{table}
$
\begin{array}{r | c c c c c c c c c c} 
n 			& min_{\rho \in \mathcal R_n, b \in [1]^n }
\textsc{REL}^{\textrm{proj}_2}_n\left( \rho,b,1 \right) & \textrm{Num}	&  b   & \rho     \\
\hline 
\hline
3			&		1+\sqrt{3}					&	2.73205		&	\{0, 1, 0\}	&	(1,2,3)		&				\\
4			&		1+3/\sqrt{2}					&	3.12132		&	\{1, 0, 0, 1\}	&	(1,2,3,4)		&				\\
5			&		1+\sqrt{25 + 2 \sqrt{5}}/2	&	3.71441		&	\{0, 1, 0, 1, 1\}	& (1, 3, 2, 4, 5)	&				\\
6			&		3 + \sqrt3/2				&	3.86603		&	\{1, 0, 0, 1, 0, 1\}	&	(1, 2, 3, 5, 4, 6)	&				\\
7			&		1+\coss{3\pi/14} + 5\sinn{\pi/7}		&	3.95125				&	\{1, 0, 0, 1, 0, 1, 0\}	&	(1, 2, 3, 7, 4, 6, 5)		&				\\
8			&		1+\sqrt2/2 + 6\sinn{\pi/8}				&	4.00321				&	\{1, 0, 0, 1, 0, 1, 0, 1\}	&	(1, 2, 3, 8, 4, 6, 5, 7)		&				\\
9			&		1+\sqrt3/2 + \coss{\pi/18} + 4\sinn{\pi/9}		& 4.21891	&	\{1, 0, 0, 1, 0, 1, 0, 1, 0\}	& (1, 2, 3, 9, 4, 8, 5, 7, 6)	&		
\end{array}
$
\caption{Lower bounds to \wsngon{\textrm{proj}_2}{n}, for $n=3,\ldots, 9$ are given in the first column as the solutions to the Linear Programs $\textsc{REL}^{\textrm{proj}_2}_n\left( \rho,b,1 \right)$. 
The second column is the numerical value of the lower bound, while the last two columns contain the corresponding minimizer permutation $\rho \in \mathcal R_n$ and binary string $b \in [1]^n$. 
}
\label{tab: n-gone lower bounds}
\end{table}
From Table~\ref{tab: n-gone lower bounds}, and for $n=9$, we derive a new lower bound to 
\wsdisk{\textrm{proj}_2}, which is 
$$
\frac{\pi}9+
\min_{\rho \in \mathcal R_9, b \in [1]^9 }
\textsc{REL}^{\textrm{proj}_2}_9\left( \rho,b,1 \right)
=
\frac{\pi}9
+
1+\sqrt3/2 + \coss{\pi/18} + 4\sinn{\pi/9} \approx 4.56798,
$$
as promised. 
The lower bounds for $n=3,4,5$ are only derived for comparison (and to illustrate how easily our methodology can derive new lower bounds to \wsngon{\textrm{proj}_2}{n}).
We also emphasize that the reported lower bound for $n=6$ agrees with the one reported and proven in~\cite{czyzowicz2020priority123}.
In particular, this implies that for $n=6$, and due to the matching upper bound for \wsngon{\textrm{proj}_2}{6} in~\cite{czyzowicz2020priority123}, we have that
relaxation $\textsc{REL}^{\textrm{proj}_2}_6\left( \rho,b,1 \right)$
to $\textsc{NLP}^{\textrm{proj}_2}_6\left( \rho,b,1 \right)$
is exact, at least for the optimizers $\rho,b$ on Table~\ref{tab: n-gone lower bounds}. 
This is no longer the case for $n=7,8,9$, since the solution found for $\textsc{REL}^{\textrm{proj}_2}_7\left( \rho,b,1 \right)$ is not embeddable in $(\reals^2, \ell_2)$. 
For this reason, it is no surprise that quantities
$$
\frac{\pi}n+
\min_{\rho \in \mathcal R_n, b \in [1]^n }
\textsc{REL}^{\textrm{proj}_2}_n\left( \rho,b,1 \right)
$$
are not increasing in $n$, and only the lower bounds for $n=7,9$ improve upon the one derived for $n=6$.
Finally, the computations for $n=10$ are more intense, but unfortunately they do not give rise to further improvement (without exhausting all permutations and binary strings, a solution was found that gives a weaker lower bound that the case $n=9$). 
Finally, the case $n=11$ cannot be treated exhaustively since the number of configurations (approximately $7\cdot10^8$) becomes forbidden, especially for the size of the linear programs that is growing too, with $\Theta(n^3)$ constraints and $\Theta(n^2)$ variables.

\section{Weighted Group Search on the Disk}
\label{sec: arithmetic weighted disk}

In this section we derive upper and lower bounds for \wsdisk{g_w}, stemming from the cost function $g_w(x,y)=wx+y$,  $w\in [0,1]$. Our upper and lower bounds are quantified by concrete formulas. Our lower bounds are derived by utilizing Lemma~\ref{lem: lb to disk}. In Figure~\ref{fig: upper lower arithmetic} we summarize our results graphically before we give the technical details in Sections~\ref{sec: upper arithmetic weighted disk},~\ref{sec: lower arithmetic weighted disk}. 

\ignore{
 \begin{figure}[h!]
\centering
  \includegraphics[width=6cm]{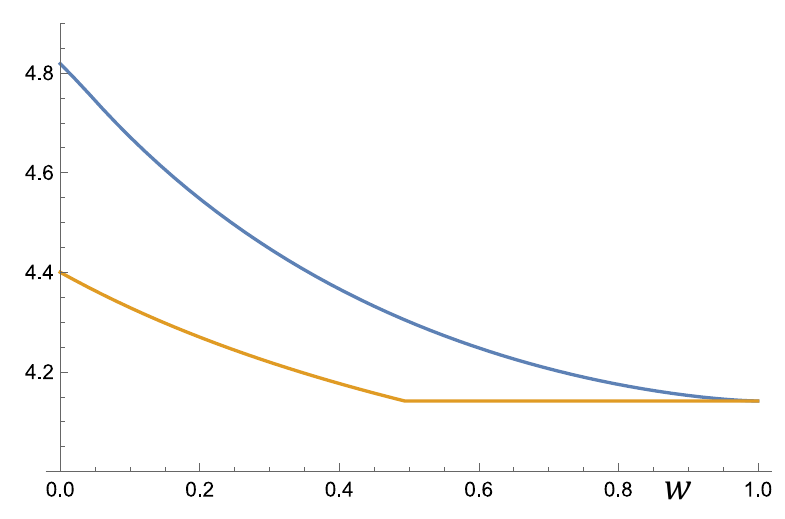}
\caption{Derived upper and lower bounds for \wsdisk{g_w}, where $g_w(x,y)=wx+y, w\in [0,1]$. }
\label{fig: upper lower arithmetic}
\end{figure}
}

\begin{figure}[h!]
  \begin{minipage}{0.45\textwidth}
    \centering
    \includegraphics[width=6cm]{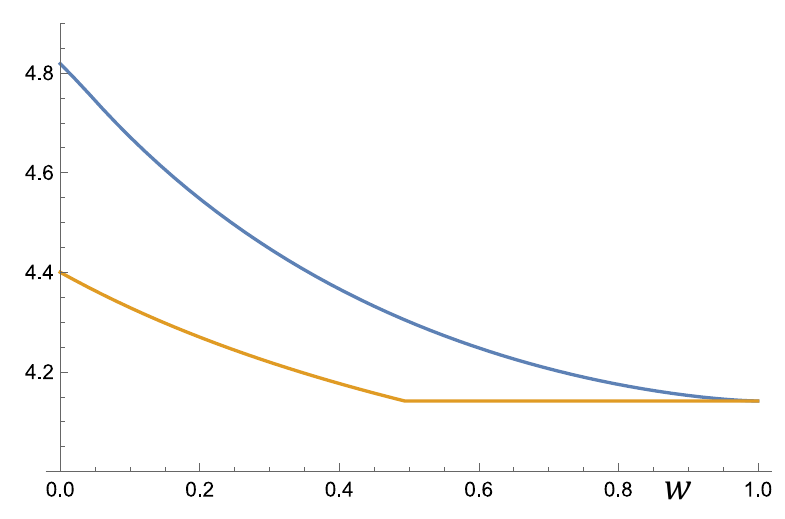}
    \caption{Derived upper and lower bounds for \wsdisk{g_w}, where $g_w(x,y)=wx+y, w\in [0,1]$.}
    \label{fig: upper lower arithmetic}
  \end{minipage}%
\hspace{.5cm} 
  \begin{minipage}{0.5\textwidth}
We emphasize that for $w=0$ we have that $g_w=g_0=\textrm{proj}_2$. 
In particular, the upper bound depicted in Figure~\ref{fig: upper lower arithmetic} is also the one reported in~\cite{czyzowicz2020priority123} for \wsdisk{\textrm{proj}_2}, while the depicted lower bound is an improvement of the one derived in~\cite{czyzowicz2020priority123}, but lower than the one we proved for the special problem \wsdisk{g_0} in Section~\ref{sec: implied LB priority evacuation}. 
This is because for the lower bounds to \wsdisk{g_w} we relied on the lower bounds we managed to prove for \wsngon{g_w}{7}. Indeed, dealing with lower bounds to \wsngon{g_w}{9} was computationally too demanding, taking into consideration that parameter $w$ also ranges in $[0,1]$. 
  \end{minipage}
\end{figure}

\subsection{Upper Bounds to Weighted Group Search on the Disk}
\label{sec: upper arithmetic weighted disk}

The upper bound results are quantified in the following lemma (and are depicted in Figure~\ref{fig: upper lower arithmetic}). 

\begin{lemma}
\label{lem: arithmetic upper bound}
For each $w \in [0,0.2]$, let $\alpha=\alpha_w$ and $\beta=\beta_w$ be the solutions to non-linear system\footnote{$\alpha_w,\beta_w$ are to be invoked only for smaller values of $w$. Also $
\gamma_1 = \gamma_2 = \gamma_3,
$ admits a solution also for higher values, but not for all $w\in [0,1]$. The value of $0.2$ was chosen only for aesthetic reasons.} 
$
\gamma_1 = \gamma_2 = \gamma_3,
$
where
\begin{align*}
\gamma_1 &= \alpha +\frac{2 \sin (\alpha )}{w+1} \\
\gamma_2 &= 2 \pi-\alpha -\beta +\frac{2 \sin \left(\alpha +\frac{\beta }{2}+\sin \left(\frac{\beta }{2}\right)\right)}{w+1} \\
\gamma_3 &=
2 \sin \left(\frac{\beta }{2}\right)+\frac{\alpha +\beta +w (-\alpha -\beta +2 \pi )}{w+1},
\end{align*}
see also Figure~\ref{fig: abParameters Arithmetic}. 
Then for each $w\in [0,1]$, \wsdisk{g_w} admits a solution with search cost at most 
$
1+ d_w + \frac{2\sinn{d_w}}{w+1}
$, 
where $d_w =\alpha_w$ if $w\leq w_0$ and $d_w=\arccoss{-\frac{w+1}{2}}$ otherwise. 
\ignore{
$$
1+ d_w + \frac{2\sinn{d_w}}{w+1}
,~\textrm{where}~~
d_w =
\left\{
\begin{array}{ll}
\alpha_w & w\leq w_0	\\ 
\arccoss{-\frac{w+1}{2}} &  w> w_0
\end{array}
\right.,
$$
and 
}
Moreover, $w_0$ is defined by equation $\alpha_w=\arccoss{-\frac{w+1}{2}}$, with $w_0 \approx 0.0456911$. 
\end{lemma}
As we will see next, the threshold of $w\leq 0.0456911$ corresponds to the critical ratio of the two weights $1/0.0456911\approx 21.8861$ that indicates when the agent with the higher weight has incentive to deviate from the search in order to expedite her evacuation. When the agents' weight ratio is less than $21.8861$, then a plain vanilla algorithm is the best we can report, which however is optimal when the weight ratio is $1$, i.e. when $w=1$, see Lemma~\ref{lem: weak g0 lb}.

\ignore{
 \begin{figure}[h!]
\centering
  \includegraphics[width=5cm]{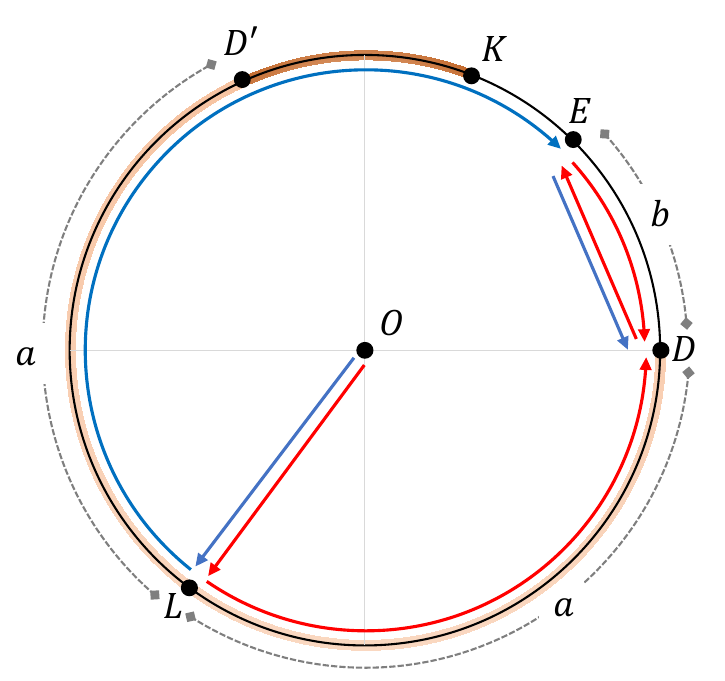}
\caption{The $(a,b)$-detour Algorithm}
\label{fig: ab detour algorithm}
\end{figure}
}

\begin{figure}[h!]
\begin{subfigure}[t]{0.45\textwidth}
\centering
\includegraphics[width=4.5cm]{figs/abAlgorithm.pdf}
\caption{The $(a,b)$-Detour Algorithm for \wsdisk{f}. Agent $A_0$ trajectory is depicted in blue, and agent $A_1$ trajectory is depicted in red.}
\label{fig: ab detour algorithm}
\end{subfigure}\hfill
\begin{subfigure}[t]{0.47\textwidth }
\includegraphics[width=6cm]{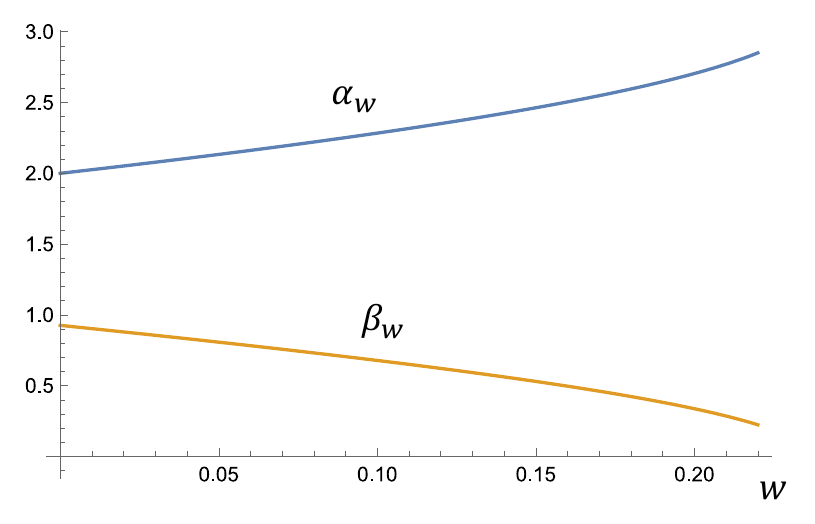}
\caption{The bevavior of parameters $\alpha_w,\beta_w$, solutions to the non-linear system of Lemma~\ref{lem: arithmetic upper bound}.
Only the values for $w\leq 0.0456911$ are relevant to the $(a,b)$-Detour Algorithm when applied to \wsdisk{g_w}.}
\label{fig: abParameters Arithmetic}
\end{subfigure}
\caption{The Detour Algorithm for \wsdisk{f}, and some of its parameters.}
\end{figure}

First we present the so-called $(a,b)$-detour Algorithm, introduced in~\cite{czyzowicz2020priority123}, specifically for \wsdisk{\textrm{proj}_2}. For this, we present the two agent trajectories $\tau_0,\tau_1: \reals_{\geq 0}\mapsto \reals^2$, where the functions depend on parameters $a,b$. Both trajectories will be piece-wise movements along arcs and chords of the unit radius disk. The description of the trajectories is given under the assumption that no target is reported or found. Should the target be located by the agent, then the agent halts and transmits the finding message to her peer. Should the finding message reach an agent, then the agent moves along the shortest path towards the target.
We emphasize that the algorithm is applicable to the search domain \textsc{Disk}, and therefore provides some solution to \wsdisk{f}, for any cost function $f$. 
For the exposition, we also use abbreviation $c(t):= \left( \coss{t}, \sinn{t} \right)$, which is the parametric equation of the unit-radius disk. Next we give formal description of the search algorithm; the reader may also consult Figure~\ref{fig: ab detour algorithm}. All movements are always at unit speed. 

\begin{definition}[The $(a,b)$-Detour Algorithm for \wsdisk{f}] ~\\
Trajectory $\tau_0: \reals_{\geq 0}\mapsto \reals^2$ of agent $A_0$, starting from origin $O$: 
Move to point $L=c(-a)$; Search clockwise up to point $E=c(b)$; Move to $D=c(0)$ along chord $ED$. \\
Trajectory $\tau_1: \reals_{\geq 0}\mapsto \reals^2$ of agent $A_1$, starting from origin $O$: 
Move to point $L=c(-a)$; Search counter-clockwise up to point $D=c(0)$; Move to $E=c(b)$ along chord $ED$; Search 
clockwise up to $D=c(0)$. 
\end{definition}

The following lemma effectively proves our upper bound claim of Lemma~\ref{lem: arithmetic upper bound}.

\begin{lemma}
\label{lem: performance of ab algo for arithmetic}
For each $w\in [0,1]$, let 
$$
a_w = \left\{
\begin{array}{ll}
\alpha_w & w \leq w_0  \\
\pi & w > w_0 
\end{array}
\right.,~~
b_w = \left\{
\begin{array}{ll}
\beta_w & w \leq w_0  \\
0 & w > w_0 
\end{array}
\right.,
$$
where $\alpha_w,\beta_w$ and $w_0\approx  0.0456911$ are as in Lemma~\ref{lem: arithmetic upper bound}. 
Then, the $(a_w,b_w)$-Detour Algorithm for \wsdisk{g_w} has evacuation cost 
$
1+ d_w + \frac{2\sinn{d_w}}{w+1}
,$ 
where 
$$
d_w =
\left\{
\begin{array}{ll}
\alpha_w & w\leq w_0	\\ 
\arccoss{-\frac{w+1}{2}} &  w> w_0
\end{array}
\right..
$$
\end{lemma}

\begin{proof}
We begin the proof by analyzing the $(\pi,0)$-Detour Algorithm that is applicable to \wsdisk{g_w} for $w\geq w_0 \approx  0.0456911$ (we shall see soon that there is continuity of the evacuation cost at $w=w_0$). 
In that case, each agent searches half the perimeter of the unit disk. If the target is reported after searching for time $t \leq \pi$, then the termination time $T_i(t)$ of the finder $i \in [1]$ is $T_i(t)=1+t$, and as per~\eqref{equa: last jump} that of the non-finder is $T_{1-i}(t)=1+t+2\sinn{t}$. 
Any of the two agents could be the finder, but for cost function $g_w(x,y)=wx+y$ with $w\in [0,1]$, it is not less costly when the finder is agent $A_0$, and for this reason, the search cost of the algorithm, should the exit be reported after searching time $t$, is 
$$
\frac{g_w\left(1+t,1+t+ 2\sinn{t} \right) }{g_w(1,1)} 
= 
1+t+\frac{2\sinn{t}}{w+1}.
$$
The last expression is increasing up to $t_w = \arccoss{-(w+1)/2}$, and then decreasing, i.e. it is concave. 
It follows that its maximum becomes $1+ t_w + \frac{2\sinn{t_w}}{w+1}$, as claimed. 

We now turn our attention to the performance of the $(\alpha_w,\beta_w)$-Detour Algorithm, for $w\leq w_0$. 
Parameter values for $\alpha_w,\beta_w$ are depicted in Figure~\ref{fig: abParameters Arithmetic}. We distinguish 4 cases as to where the target may be placed, and in each of them we evaluate the termination times of the agents, along with the associated worst case search costs $c_1, c_2, c_3, c_4$ (and our parameters $\alpha_w, \beta_w$ are identified by making all costs equal). 
For the analysis, the reader may also consult Figure~\ref{fig: ab detour algorithm}. 
In what follows arcs are read counter-clockwise. 

\textbf{Case 1:} Target is placed in arc $DD'$, i.e. it is found while the two agents search in opposite direction for up to $\alpha_w$ time. The search in this phase lasts for time $\alpha_w$, hence the target is reported after searching for $t \leq \alpha_w$. From the analysis we performed when $w\leq w_0$, we know that in the current case, the search cost is $1+t+\frac{2\sinn{t}}{w+1}$. For all $w\leq w_0$ we have that $\alpha_w \leq \arccoss{(-(w+1)/2)}$, see also Figure~\ref{fig: awarccosComparison} (in fact $w_0$ was identified by equation  $\alpha_{w_0} \leq \arccoss{(-(w_0+1)/2)}$, see statement of Lemma~\ref{lem: arithmetic upper bound}). 
Therefore, as per our previous analysis, the cost in this case remains strictly increasing in $t$, and hence in Case 1, the worst case search cost is 
$$c_1 = 1+ \alpha_w + \frac{2\sinn{\alpha_w}}{w+1}.$$

\begin{figure}[htb!]
\begin{subfigure}[t]{0.45\textwidth}
\centering     
\includegraphics[width=5cm]{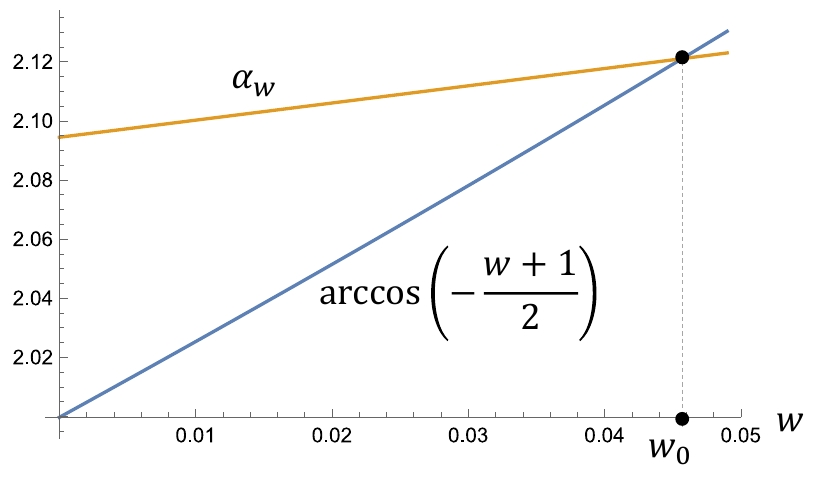}
\caption{Comparison between $\alpha_w$ and $\arccos{-(w+1)/2}$, and the defining value of $w_0$.}
\label{fig: awarccosComparison}
\end{subfigure}~~~~~~~~~~ \hfill
\begin{subfigure}[t]{0.6\textwidth}
\centering     
\includegraphics[width=4.5cm]{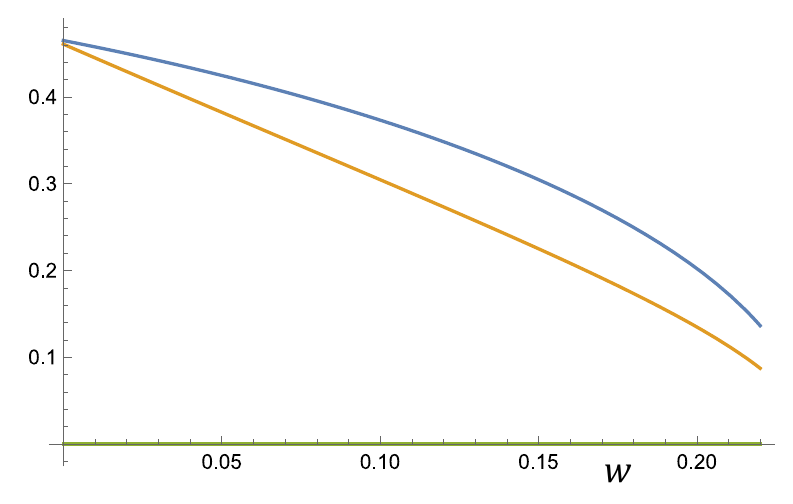} 
\caption{The plot of $2\pi - 2\alpha_w - \beta_w - 2\sinn{\beta_w}/2$ in blue proves that for all $w\leq w_0$, 
point $K$ in Figure~\ref{fig: ab detour algorithm} lies within arc $D'E$ (case 2 in the proof of Lemma~\ref{lem: performance of ab algo for arithmetic}). 
The plot of $-2\pi + 2\alpha_w + 2\beta_w + 2\sinn{\beta_w}/2$ in orange proves that for all $w\leq w_0$, 
agent $A_0$ arrives at $E$ before agent $A_1$ reaches $D$ for the second time, (case 4 in the proof of Lemma~\ref{lem: performance of ab algo for arithmetic}). The line $y=0$ in green is depicted for reference. 
}
\label{fig: case 2 arrivals}
\end{subfigure}\hfill
\begin{subfigure}[t]{0.45\textwidth}
\centering     
\includegraphics[width=5cm]{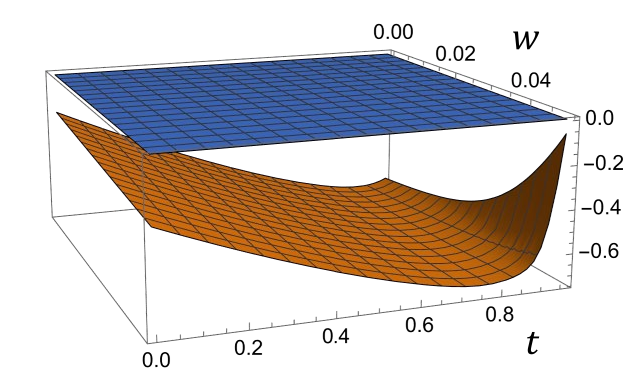}
\caption{The derivative of search cost~\eqref{equa: case2 monotone}. 
The plot depicts the plane $z=0$, showing that the derivative of the function is negative and bounded away from $0$, in the subject domain.
}
\label{fig: case2Monotonicity}
\end{subfigure}\hfill
\begin{subfigure}[t]{0.45\textwidth}
\centering     
\includegraphics[width=5cm]{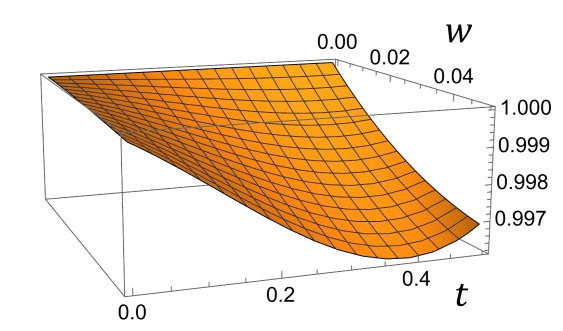}
\caption{The derivative of search cost~\eqref{equa: case4 monotone}. 
The values are are close to $1$ showing that the cost function is strictly increasing in the subject domain.
}
\label{fig: case4Monotonicity}
\end{subfigure}\hfill
\caption{The behavior of some expressions, as needed in the proof of Lemma~\ref{lem: performance of ab algo for arithmetic}.} 
\end{figure}

\ignore{
\subfigure[The derivative of search cost~\eqref{equa: case2 monotone}. 
The plot depicts the plane $z=0$, showing that the derivative of the function is negative and bounded away from 0, in the domain $t \in [0, 2\sinn{\beta_w/2}]$, for all $w\in [0,w_0]$. By Figure~\ref{fig: abParameters Arithmetic}, we know that 
$\beta_w\leq 0.925793$ (the value used in~\cite{czyzowicz2020priority123} for specifically for \wsdisk{g_0}). For this reason, our plot that ranges over $w\in [0,2\sinn{1/2}]$ is sufficient for our purposes. 
]
{\label{fig: case2Monotonicity}
\includegraphics[width=5cm]{figs/case2Monotonicity.pdf}}
~~\hspace{1cm}~~~~~~~
}

\ignore{
\subfigure[........... ]
{\label{fig: awarccosComparison}
\includegraphics[width=5cm]{figs/awarccosComparison.pdf}}
~~\hspace{1cm}~~~~~~~
}

\textbf{Case 2:} Target is placed in arc $D'K$, i.e. while agent $A_1$ moves from $D$ to $E$ along chord $DE$. \\
First we show that point $K$ lies in between arc $D'E$, that is agent $A_1$ arrives at $E$ before agent $A_0$. To see why, notice that agent $A_0$ arrives at $E$ in time 
$
1+2\pi -\alpha_w -\beta_w
$.
Similarly, agent $A_1$ arrives at $E$ in time $1+\alpha_w + 2\sinn{\beta_w}/2$. The time gap of these events is 
$$
1+2\pi -\alpha_w -\beta_w
\left( 
1+\alpha_w + 2\sinn{\beta_w}/2
\right)
=
2\pi - 2\alpha_w - \beta_w - 2\sinn{\beta_w}/2.
$$
The latter expression remains non-negative for all $w\leq w_0$ as claimed, see Figure~\ref{fig: case 2 arrivals}. 

To resume, we verified that for the duration agent $A_0$ searches arc $D'K$ (of length $2\sinn{\beta_w/2}$), agent $A_1$ moves along chord $DE$ (of the same length). If the exit is reported by $A_0$ at time $t \in [0, 2\sinn{\beta_w/2}]$, then the termination cost of that agent is $T_0(t)=1+\alpha_w+t$. Note that in that case, the target is placed at $c(-2\alpha_w-t)$. 
Agent $A_1$ needs to additionally spend the time to reach the target, and this is what we compute next. For this we utilize the parametric equation for a unit speed movement between points $A,B$, which reads as 
$$
l\left(t,A,B\right):=A+\frac{t}{\norm{B-A}}(B-A).
$$
More specifically, for the aforementioned placement of the target, agent $A_1$ lies at 
$$l\left(t,D,E\right)=l\left(t,c(0),c(\beta_w)\right),$$
 and hence its termination cost is, as also per~\eqref{equa: last jump},
$$
T_1(t)=1+\alpha_w+t+\norm{c(-2\alpha_w-t) - l\left(t,c(0),c(\beta_w)\right) }.
$$
We conclude that in this case, the search cost for cost function $g_w$ is 
\begin{equation}
\label{equa: case2 monotone}
\frac{g_w(T_0(t), T_1(t)) }{g_w(1,1)}
=
\frac{g_w(T_0(t), T_1(t)) }{w+1}
=
1+\alpha_w+t
+ \frac{\norm{c(-2\alpha_w-t) - l\left(t,c(0),c(\beta_w)\right) }}{w+1}.
\end{equation}
The latter function admits an analytic formula, whose derivative for $t \leq 2\sinn{\beta_w/2}$ is depicted in Figure~\ref{fig: case2Monotonicity}, together with plane $z=0$. 
The plot shows that the derivative of the function is negative and bounded away from $0$, in the domain $t \in [0, 2\sinn{\beta_w/2}]$, for all $w\in [0,w_0]$. By Figure~\ref{fig: abParameters Arithmetic}, we know that 
$\beta_w\leq \beta_0 =0.925793$ (value $\beta_0$ was used in~\cite{czyzowicz2020priority123} specifically for \wsdisk{g_0}). For this reason, our plot ranging over $w\in [0,2\sinn{1/2}]$ is sufficient for the purpose of showing that  the search cost in decreasing in $t$. Since at $t=0$, the case was already considered in Case 1, we conclude that for the worst case search cost in this case we have $c_2=c_1$.

\textbf{Case 3:} Target is placed in arc $KE$. \\
Note that by the previous case, while agent $A_0$ is moving counter-clockwise along chord $KE$, agent $A_1$ is moving counter-clockwise along arc $ED$, starting from points $K,E$, respectively. 
First we argue that agent $A_0$ reaches $E$ before $A_1$ reaches $D$ for the second time. 
Indeed, agent $A_0$ reaches $E$ in time $1+2\pi - 2\alpha_w - \beta_w$, 
while $A_1$ reaches $D$ for the second time at 
$1+\alpha_w + 2\sinn{\beta_w/2}+\beta_w$. So we have 
$$
1+\alpha_w + 2\sinn{\beta_w/2}+\beta_w
-
(
1+2\pi - 2\alpha_w - \beta_w
)
=
- 2\pi + 2\alpha_w +2 \beta_w + 2\sinn{\beta_w/2}.
$$
The latter expression is depicted in Figure~\ref{fig: case 2 arrivals} and is positive. 

This means that Case 3 has a time span equal to the length of arc $KE$ which equals 
$
2\pi - 2\alpha_w - \beta_w - 2\sinn{\beta_w/2}
$
and during that time, the distance between the two agents is invariant (and equal to the length of arc $KE$). 
Therefore the search cost in this case is the highest when the target is found by $A_0$ as late as possible, i.e. close to $E$ in arc $KE$ (the maximum is not attained, rather we calculate the supremum of the search cost). 
In the limit, the termination cost of agent $A_0$ is $T_0=1+2\pi - \alpha_w - \beta_w$, while agent pays in addition its distance to the exit, 
as in~\eqref{equa: last jump},
that is 
$$
T_1 = 1+2\pi - \alpha_w - \beta_w + 2\sinn{\alpha_w + \beta_w/2 + \sinn{\beta_w/2}}.
$$ 
Overall, the worst case termination cost in this case is 
$$
c_3 
=
\frac{g_w(T_0, T_1) }{g_w(1,1)}
=
\frac{g_w(T_0, T_1) }{w+1}
=
1+2\pi - \alpha_w - \beta_w
+
\frac{2\sinn{\alpha_w + \beta_w/2 + \sinn{\beta_w/2}}}{w+1}
$$

\textbf{Case 4:} Target is placed in arc $DE$. \\
In this case, the target is found by agent $A_1$. The phase has time span equal to the length of arc $ED$, i.e. $2\sinn{\beta_w/}$, and recall that from our previous analysis, agent $A_1$ reaches $D$ for the second time while agent $A_0$ is still moving along chord $ED$ from $E$ towards $D$. 

Agent $A_1$ starts searching arc $ED$ when $A_0$ is at point $K$ still moving counter-clockwise along the perimeter. Agent $A_0$ reaches point $E$ in time equal to the length of arc $KE$, i.e. 
$
2\pi - 2\alpha_w - \beta_w - 2\sinn{\beta_w/2}.
$
In this time window, the distance of the two agents is invariant. In fact, the termination cost in this case is dominated by the cost of Case 3, since now the extra distance towards the target needs to be covered by the ``light'' agent whose termination cost is scaled by $w\in [0,1]$. For this reason, we may assume that the target is reported by $A_1$ in arc $ED$ after $A_0$ reaches $E$, and while $A_1$ moves along chord $ED$, from $E$ towards $D$. 

We reset the clock at time $1+2\pi -\alpha_w - \beta_w$ 
when $A_0$ reaches point $E$, after which time the agent moves along trajectory 
$l(t, c(\beta_w), c(0)), t\in [0,2\sinn{\beta_w/2}]$. 
When we reset the clock, agent $A_1$ lies in 
$c\left(
2\alpha_w + 2\beta_w + 2\sinn{\beta_w/2} - 2\pi
\right)
$, and therefore at additional time $t$ it is located at 
$$c\left(
2\alpha_w + 2\beta_w + 2\sinn{\beta_w/2} - 2\pi - t
\right),$$
where $t \leq 2\alpha_w + 2\beta_w + 2\sinn{\beta_w/2} - 2\pi$. 

To conclude, each $t \in [0, 2\alpha_w + 2\beta_w + 2\sinn{\beta_w/2} - 2\pi]$ corresponds to a placement of the target that is located by agent $A_1$ for her termination cost 
$
T_1(t) = 1+2\pi -\alpha_w - \beta_w + t
$.
Agent $A_0$ will have to pay the additional cost of moving to the target, and hence her termination cost in this case, as in~\eqref{equa: last jump}, 
is 
$$
T_0(t) = 1+2\pi -\alpha_w - \beta_w + t + \norm{l(t, c(\beta_w), c(0)) -  c\left(
2\alpha_w + 2\beta_w + 2\sinn{\beta_w/2} - 2\pi - t
\right)   }
.$$
Hence, in this case the termination cost for cost function $g_w$ becomes 
\ignore{
\begin{equation}
\label{equa: case4 monotone}
\frac{g_w(T_0(t), T_1(t)) }{g_w(1,1)}
=
1+2\pi -\alpha_w - \beta_w + t
+
w\frac{\norm{l(t, c(\beta_w), c(0)) -  c\left(
2\alpha_w + 2\beta_w + 2\sinn{\beta_w/2} - 2\pi - t
\right)}}{w+1}.
\end{equation}
}
\begin{eqnarray}
\label{equa: case4 monotone}
~~~~~~ \frac{g_w(T_0(t), T_1(t)) }{g_w(1,1)} \\
=&
1+2\pi -\alpha_w - \beta_w + t
+
w\frac{\norm{l(t, c(\beta_w), c(0)) -  c\left(
2\alpha_w + 2\beta_w + 2\sinn{\beta_w/2} - 2\pi - t
\right)}}{w+1}. 
\notag
\end{eqnarray}

Now we claim that this cost is increasing in $t \leq 2\alpha_w + 2\beta_w + 2\sinn{\beta_w/2} - 2\pi$. 
Note that the function admits a closed formula, and therefore we can compute its derivative, that is depicted in Figure~\ref{fig: case4Monotonicity}.
The plot confirms that derivative is remains close to $1$, as expected, since the 2 agents are moving towards the same point $c(0)$, one along the arc $ED$ and one along the chord $ED$. The domain in the figure was chosen to be $t\in [0,5]$ which is a superset of time window of the case we are considering, because $ 2\alpha_w + 2\beta_w + 2\sinn{\beta_w/2} - 2\pi \leq 
 2\alpha_0 + 2\beta_0 + 2\sinn{\beta_0/2} - 2\pi
\approx 0.460808$. 

To conclude, in this case, the termination cost is the highest when the target is placed within arc $ED$ arbitrarily close to point $D=c(0)$ (again the maximum is not attained because point $D$ was visited before, so we report the supremum). But then, the target is found by agent $A_1$ in time 
$T_1= 1+\alpha_w+2\sinn{\beta_w/2}+\beta_w$, while agent $A_0$ arrives at point $D$ at time
$$
T_0 =
1+2\pi-\alpha_w-\beta_w+2\sinn{\beta_w/2}.
$$
That is, in the current case the termination cost equals 
$$
c_4 
=
\frac{g_w(T_0, T_1) }{g_w(1,1)}
=
1+2\sinn{\beta_w/2}
+
\frac{\alpha_w +\beta_w +w (-\alpha_w -\beta_w +2 \pi )}{w+1},
$$
which also concludes the termination cost analysis for all cases 1,2,3,4. 

In order to conclude the claim of Lemma~\ref{lem: performance of ab algo for arithmetic} by calculating the termination cost of the $(\alpha_w,\beta_w)$-Detour Algorithm when $w\leq w_0$, which equals
$$
\max\{ c_1, c_2, c_3,c_4\}.
$$
Note that $c_1 = c_2$, and as per the statement of Lemma~\ref{lem: arithmetic upper bound} the values of parameters $\alpha_w,\beta_w$, as a function of $w$, where chosen specifically by requiring that $c_1=c_3=c_4$. 
Therefore, the termination cost is also described by the formula derived in Case 1, that is $1+ \alpha_w + \frac{2\sinn{\alpha_w}}{w+1}$. 
\qed\end{proof}

\subsection{Lower Bounds to Weighted Group Search on the Disk}
\label{sec: lower arithmetic weighted disk}

In this section we give the details of how we obtained the lower bound values to \wsdisk{g_w} reported in Figure~\ref{fig: upper lower arithmetic}.

\begin{theorem}
\label{thm: gw lb}
For all $w\in [0,1]$, no algorithm for \wsdisk{g_w} has evacuation cost less than \\
$\max\left\{
1+\pi,
1+\pi/7 + \frac1{w+1}\coss{3\pi/14} + 5\sinn{\pi/7}
\right\}
\approx
\max\left\{
4.14159,
3.61822+\frac{0.781831}{w+1}
\right\}
.$
\end{theorem}

We start with a weak lower bound. 
\begin{lemma}
\label{lem: weak g0 lb}
No algorithm for \wsdisk{g_w} has evacuation cost less than 
$1+\pi$.
\end{lemma}

\begin{proof}
An arbitrary algorithm for \wsdisk{g_w} needs time at least $1+\pi$ to search the entire domain. 
This means that there is always a target placement for which the termination time of each agent is at least $1+\pi - \epsilon$, resulting in overall search cost 
$
\frac{g_w(1+\pi - \epsilon,1+\pi - \epsilon)}{w+1}=1+\pi - \epsilon
.$
\qed\end{proof}
Notably, the result is tight for $w=1$. Indeed, by Lemma~\ref{lem: performance of ab algo for arithmetic} the upper bound for \wsdisk{g_1} uses the $(\pi,0)$-Detour Algorithm with performance $1+d_1+2\sinn{d_1}2 = 1+\pi$ (note that $d_1=\arccoss{-\tfrac{1+1}{2}}=\pi$). 

For small values of $w$, we obtain better bounds by deriving lower bounds to \wsngon{g_w}{n} for $n=7$. 
Unlike how we approached \wsdisk{\textrm{proj}_2} (which is the same as \wsdisk{g_0}), this time we need to solve $\textsc{REL}^{g_w}_n\left( \rho,b,1+\frac{\pi}n \right)$ not only for all permutations $\rho$ and binary strings $b$, but also for enough many values of $w \in [0,1]$. In this direction, we use $n=7$ and compute the lower bounds to \wsdisk{g_w} from $w=0$ up to $1$ with step size $0.001$\footnote{Considering $7$-\textsc{Gons} for our proof is the highest value we could handle computationally, due to the large number of $w$ values we needed to consider. For comparison, the results we obtained for \wsdisk{\textrm{proj}_2} using $9$-\textsc{Gons} took several days to be computed. The order by which the number of configurations increase between $n$-\textsc{Gons} and $(n+1)$-\textsc{Gons}, each time computing $\textsc{REL}^{g_w}_n\left( \rho,b,1+\frac{\pi}n \right)$, is $2n$, but this does not take into account the additional computation cost for solving much larger linear programs with $\Theta(n^2)$ variables and $\Theta(n^3)$ constraints.
}. 
The following is a strict generalization of the results in Table~\ref{tab: n-gone lower bounds} pertaining to \wsngon{g_w}{7} with $n=7$, and was obtained by computer assisted calculations for determining the minimum value of $\textsc{REL}^{g_w}_7\left( \rho,b,1 \right)$ over all permutations $\rho$ and binary strings $b$.

\begin{lemma}
\label{lem: gw 7gone}
No algorithm for \wsngon{g_w}{7} has evacuation cost less than 
$
1+\frac1{w+1}\coss{\tfrac{3\pi}{14}} + 5\sinn{\tfrac{\pi}{7}}.
$
\end{lemma}

Strictly speaking the bound of Lemma~\ref{lem: gw 7gone} equals $\min_{\rho \in \mathcal R_7, b \in [1]^7 }
\textsc{REL}^{g_w}_7\left( \rho,b,1 \right)$ only for $w\leq 0.8$ (for higher values, this technique gives a slightly better lower bound), however, this is already subsumed by the bound of Lemma~\ref{lem: weak g0 lb}. 
We conclude that for each $w$, and using Lemma~\ref{lem: lb to disk} and the calculations~\eqref{equa: pseudo linear}, 
no algorithm for \wsdisk{g_w} has evacuation cost less than 
$$
\frac{1}{g_w(1,1)} \min_{\rho \in \mathcal R_7, b \in [1]^7 }
\textsc{REL}^{g_w}_7\left( \rho,b,1+\frac{\pi}7 \right)
=
\frac{\pi}7+
\min_{\rho \in \mathcal R_7, b \in [1]^7 }
\textsc{REL}^{g_w}_7\left( \rho,b,1 \right),
$$
where in particular $\textsc{REL}^{g_w}_7\left( \rho,b,1 \right)$ is a lower bound to the optimal $(\rho,b)$-algorithm for \wsngon{g_w}{7}.
To resume, we solve $\textsc{REL}^{g_w}_7(\rho,b,1)$ for all permutations $\rho \in \mathcal R_7$ and all binary strings $b \in [1]^n$, and we report the smallest values in Lemma~\ref{lem: gw 7gone}. Overall, this implies that no algorithm for \wsdisk{g_w} has search cost less than 
$1+\pi/7 + \frac1{w+1}\coss{3\pi/14} + 5\sinn{\pi/7}.$
This observation, together with Lemma~\ref{lem: weak g0 lb}  give the proof of Theorem~\ref{thm: gw lb}. 
Lastly we note that the transition value of $w$ for which the lower bound of Lemma~\ref{lem: weak g0 lb} is weaker equals 
$$
\tfrac{7 \cos \left(\frac{3 \pi }{14}\right)}{6 \pi -35 \sin \left(\frac{\pi }{7}\right)}-1\approx 0.493827.
$$


\vspace{-.3cm}
\section{Conclusion}
\label{sec: conclusions}
\vspace{-.2cm}
It this work we studied the weighted group search problem on the disk with 2 agents operating in the wireless model, a problem that was previously studied on the line. 
The weighted problem on the disk is a generalization of the so-called priority evacuation problem, for which the best upper and lower bounds known exhibit a significant gap. 
For the problem on the disk, we designed and analyzed algorithms that adapt with the underlying cost function. We complemented our results by providing lower bounds for the entire spectrum of arithmetic weighted average cost functions. 
Our most significant contribution is the framework we developed in order to prove the lower bounds. More specifically, we introduced linear programs, which arise as relaxations to non-linear programs modeling the behavior of optimal solutions to combinatorial search type problems. The framework is applicable to more general cost functions, as well as it can be adjusted to provide lower bounds to highly asymmetric searchers' specifications. Among others, we demonstrate the power of this technique by improving the previously best lower bound known for priority evacuation, from $4.38962$ to $4.56798$. The bound can be further improved if one uses the provided framework and utilizes more efficient computational tools, e.g. software tailored to solving LPs (the current symbolic results were obtained using \textsc{Mathematica}). 
One of the the most challenging future directions in the area would be to consider other searchers' communication models, including the face-to-face model, that has been proved to be challenging even with standard cost functions.

\bibliographystyle{abbrv}
\bibliography{LPbasedLowerBounds}

\end{document}